\documentclass[letterpaper, 10 pt, conference]{ieeeconf}  

\IEEEoverridecommandlockouts                              
\usepackage[utf8]{inputenc}

\usepackage{amsfonts}
\usepackage{amsmath}
\usepackage{amssymb}
\usepackage{acronym}
\usepackage{xcolor}

\usepackage{mathtools}
\usepackage{paralist}
\usepackage{diagbox}

\newtheorem{theorem}{Theorem} 
\newtheorem{definition}{Definition}

\usepackage{MnSymbol,wasysym}
\usepackage{changes}

\usepackage{tikz}
\usetikzlibrary{arrows, arrows.meta}
\usetikzlibrary{automata,positioning}
\usetikzlibrary{shapes.geometric}
\usepackage{subfig}

\newtheorem{problem}{Problem}
\newtheorem{example}{Example}

\newtheorem{lemma}{Lemma}
\newtheorem{assumption}{Assumption}

\newif\ifuseboldmathops
\newif\ifuseittextabbrevs
\useboldmathopstrue   

\ifuseittextabbrevs

	\newcommand{\ie}{{\it i.e.}}

\else

	\newcommand{\ie}{i.e.~}

\fi

\ifuseboldmathops

\else

\fi

\ifuseboldmathops

\else

\fi

\ifuseboldmathops

\else

\fi

\newcommand{\obs}{O}

\newcommand{\dist}[1]{\mathcal{D}(#1)}

\newcommand{\calA}{\mathcal{A}}

\newcommand{\supp}{\mathsf{Supp}}

\newcommand{\plays}{\mathsf{Plays}}
\newcommand{\prefplays}{\mathsf{PrefPlays}}

\newcommand{\calAP}{\mathcal{AP}}
\newcommand{\winplays}{\mathsf{WPlays}}
\newcommand{\opaquewinplays}{\mathsf{OWPlays}}
\newcommand{\revealwinplays}{\mathsf{RWPlays}}

\newcommand{\win}{\mathsf{Win}}

\newcommand{\Eventually}{\lozenge \, }

\newcommand{\init}{{\iota}}

\acrodef{mdp}[MDP]{Markov decision process}

\acrodef{asw}[ASW]{Almost-Sure Winning}
\acrodef{ltlf}[LTL$_f$]{Linear Temporal Logic over Finite Traces}
\acrodef{ltl}[LTL]{Linear Temporal Logic}
\acrodef{scltl}[co-safe LTL]{syntactically co-safe Linear Temporal Logic}
\acrodef{dfa}[DFA]{Deterministic Finite Automaton}
\acrodef{des}[DES]{Discrete Event System}

\title{\LARGE \bf
 Synthesis of Opacity-Enforcing Winning Strategies Against Colluded Opponent}

\author{Chongyang Shi, Abhishek N. Kulkarni, Hazhar Rahmani, and Jie Fu$^\ast$
\thanks{The authors are with the Department of Electrical and Computer
Engineering, University of Florida, Gainesville, FL 32605, USA.
        {\tt\small \{c.shi, a.kulkarni2, h.rahmani, fujie\}@ufl.edu }.}%
 \thanks{$\ast$: the corresponding author.}
}

\begin{document}

\maketitle
\begin{abstract}
    This paper studies a language-based opacity enforcement in a two-player, zero-sum game on a graph. In this game, player 1 (P1) wins if it can achieve a secret temporal goal described by the language of a finite automaton, no matter what strategy the opponent player 2 (P2) selects. In addition, P1 aims to win while making its goal opaque to a passive observer with imperfect information.
    However, P2 colludes with the observer to reveal P1's secret whenever P2 cannot prevent P1 from achieving its goal, and therefore, opacity must be enforced against P2. We show that a winning and opacity-enforcing strategy for P1 can be computed by reducing the problem to solving a reachability game augmented with observer's belief states. Furthermore, if such a strategy does not exist, winning for P1 must entail the price of revealing his secret to the observer. We demonstrate our game-theoretic solution of opacity-enforcement control through a small illustrative example and 
    in a robot motion planning problem.
    %
\end{abstract}
 
\section{Introduction}

Non-interference or opacity is a security and privacy property that evaluates whether an observer (intruder) can infer a secret of a system by observing its behavior. The secret could be a language generated by the system (language-based opacity)~\cite{dubreil2008opacity,lin2011opacity,wu2013comparative}; a system state---such as initial, current, or final state---(state-based opacity)~\cite{bryans2005modelling,bryans2008opacity,saboori2007notions,sabooriCurrentStateOpacityFormulations2014,wu2013comparative}; or some past state ($k$-step and infinite step opacity)~\cite{saboori2007notions,saboori2009verification,wu2013comparative}.   
Recent work by Wintenberg \emph{et al.}~\cite{wintenberg2022general} unifies those notions of opacity for discrete event systems and provides transformations from state-based notions of opacity to language-based notions.
Enforcing opacity using supervisory control has been extensively studied (see the survey by Lafortune \emph{et al.}~\cite{lafortuneHistoryDiagnosabilityOpacity2018}, the overview by Jacob \emph{et al.}~\cite{jacob2016overview}, and chapters in the book by Hadjicostis~\cite{hadjicostis2020estimation}).  Besides discrete event system model and supervisory control approach, 
  opacity verification and enforcement in Petri net systems have been studied \cite{bryans2005modelling,bryans2008opacity} and optimization-based approach  are developed to enforce opacity in Petri net systems \cite{congOnlineVerificationCurrentstate2018,basileNoninterferenceEnforcementSupervisory2021,basileOptimizationbasedApproachAssess2022}.

This paper investigates a game-theoretic approach to enforce language-based opacity against an opponent who colludes with the passive observer. 
Our problem formulation is motivated by security applications of cyber-physical systems and robotics: Consider a robot (referred to as player 1/P1) aims to accomplish a high-confidential task in a dynamic environment. Besides task complition, the robot must ensure a passive observer cannot infer if the task is accomplished, based on the observer's partial information. However, there are uncontrollable events or other agents in the operational environment. 
In the worst case scenario, the uncontrollable environment (referred to as player 2/P2) may act in a way that forces the robot to reveal its secret. The question is, does the robot have an opacity-enforcing strategy to satisfy the secret surely, against a colluding P2?

The connection between reactive synthesis and supervisory control 
has been  studied by Rüdiger \emph{et al.}~\cite{ehlersSupervisoryControlReactive2017}, who show the non-blocking supervisory control problem can be reduced to reactive synthesis of a maximal permissive strategy. In light of this connection, we aim  to establish another connection of opacity-enforcing supervisory control with reactive synthesis using solutions of games on graphs \cite{zielonkaInfiniteGamesFinitely1998}.
Our game formulation is closely related to that by Hélouët \emph{et al.}~\cite{helouetOpacityPowerfulAttackers2018}, who develop a game-theoretic controller that enforces state-based opacity against different types of attackers. They employ a turn-based safety game where the goal of the controller is to avoid reaching a state at which the attacker's belief is a subset of the secret states. The safety objective   can be satisfied in two ways: Either a path/execution  never reaches a secret state, or a path that reaches a secret state but is observation-equivalent to a path that does not reach a secret state.
  Different from  \cite{helouetOpacityPowerfulAttackers2018}, in our game,  P1 aims to ensure satisfying the secret (a language accepted by an automaton) in an uncontrollable environment modeled as player 2,  while enforcing opacity from a passive observer.  We show the opacity-enforcing control can be formulated as a reachability game whose states augments the observer's belief. 
  Based on the determinancy of turn-based games \cite{gradelAutomataLogicsInfinite2002,martinBorelDeterminacy1975}, we prove  that if, for a given state, P1 can satisfy the secret from but cannot enforce its opacity with probability one, then for any secret-satisfying(called winning) strategy of P1,  any counter-strategy of P2 will surely reveal the secret of P1 to the observer. 
  This fact holds regardless if P1 or P2 uses a deterministic/non-deterministic/randomized   strategy.
We illustrate the proposed synthesis method using a small game on graph example and a robot motion planning problem in an adversarial environment.

\section{Preliminaries and Problem Formulation}
\subsection{Preliminaries}
\paragraph*{Notations} Let $X$ be a finite set. We use $\dist{X}$ to denote the set of all probability distributions over $X$. 
For any distribution $d \in \dist{X}$, the set of elements in $X$ with non-zero probability under $d$ is denoted $\supp(d)$, that is, $\supp(d)=\{x\in X\mid d(x)>0\}$. 


We model the interaction between a controllable agent, player 1(P1), and an uncontrollable environment, player 2(P2) using a turn-based game arena defined as follows.
\begin{definition}
\label{def:arena}
The \emph{transition system} (arena) of a two-player, turn-based, deterministic game on a graph is a tuple
\[
G =  \langle S := S_1 \cup S_2, A := A_1 \cup A_2, T, s_0, \calAP, L  \rangle
\]
in which
\begin{itemize}
        \item $S = S_1 \cup S_2$ is a finite set of  states partitioned into $S_1$ and $S_2$ where at a state in $S_1$, P1 chooses an action, and at a state in $S_2$, P2 selects an action;
        
        \item $A = A_1 \cup A_2$ is the set of actions, where $A_1$ (resp., $A_2$) is the set of actions for P1 (resp., P2); 
        
        \item $T : (S_1 \times A_1)\cup (S_2 \times A_2) \rightarrow S$ is a \emph{deterministic} transition function that maps a state-action pair $(s, a) \in S_i \times A_i$ to a next state $s' \in S$ for all $i \in \{1, 2\}$;
        
        \item $s_0 \in S$ is the initial state; 
        \item $\calAP$ is the set of atomic propositions; and
        \item $L: S\rightarrow 2^{\calAP}$ is the labeling function that maps a state to a set of atomic propositions that evaluate true at that state. 
        \end{itemize}
\end{definition}

A finite play $\rho = s_0 a_0 s_1 a_1 s_2\ldots s_n $ is a sequence of interleaving states and players' actions such that $T(s_i,a_i) = s_{i+1}$ for all integers $0 \leq i \leq n-1$. 
The labeling of the play $\rho$, denoted $L(\rho)$, is defined as $L(\rho) = L(s_0)L(s_1)\ldots L(s_n)$. That is, the labeling function omits the actions from the play and applies to states only. 
We use $\plays(G) \subseteq (S\times A)^\ast S$ to denote the set of finite plays that can be generated from the game $G$ and we use $\prefplays(G)$ to denote the set of prefixes\footnote{A word $u$ is a prefix of word $w$ if and only if there exists word $v$ that $w=u\cdot v$, where $\cdot$ is the symbol for concatenation.} of the finite plays within $\plays(G)$.
A  randomized (resp. deterministic) strategy of player $i$ is  a function $\pi_i: \prefplays(G)\rightarrow \dist{A_i}$ (resp. $\pi_i: \prefplays(G)\rightarrow {A_i}$) that maps a prefix/history of a play into a distribution over actions (resp. a single action for player $i$).
In our notation, $\plays(\rho, \pi_1,\pi_2) $ is the set of possible plays that can be generated (with a non-zero probability) when P1 and P2 follow the strategy profile $(\pi_1, \pi_2)$ starting from the play $\rho$.
Formally, $\rho' \in \plays(\rho,\pi_1,\pi_2) $ if and only if $\rho  $ is a prefix of $\rho'$, \ie, $\rho' =\rho \cdot s_ka_{k}s_{k+1}\ldots s_n$ where if $s_j \in S_i$ then $a_j\in \supp(\pi_i(\rho \cdot s_k a_k\ldots s_j))$ and $T(s_j,a_j) = s_{j+1}$, for all $k \leq j < n$ and $i\in \{1,2\}$.
%
%
%
For $i \in \{1, 2 \}$, the set of all randomized strategies for player $i$ is denoted $\Pi_i$.

\paragraph{P1' temporal objective}
In the game arena, Player P1 (pronouns he/him/his) intends to achieve a temporal objective $\varphi \subseteq (2^\calAP)^\ast$, while Player P2 (pronouns she/her) aims to prevent P1 from achieving that objective. 

 The set of words satisfying the temporal objective $\varphi$ is equivalently represented by the language of  a \ac{dfa}\footnote{For linear temporal logic over finite traces, the formula can be represented as \ac{dfa}s. We omit the introduction of temporal logic for conciseness.}.

 \begin{definition}[\ac{dfa}]
 A \ac{dfa} is a tuple $\mathcal{A} = ( Q, \Sigma, \delta, \init, F )$ in which (1) $Q$ is the set of states; (2) $\Sigma$ is the alphabet (set of input symbols); (3) $\delta: Q\times \Sigma \rightarrow Q$ is a deterministic transition function and is complete \footnote{For any $Q\times \Sigma$, $\delta(q,\sigma)$ is defined. An incomplete transition function can be completed by adding a sink state and redirecting all undefined transitions to that sink state.}; (4) $\init$ is the initial state; and (5) $F \subseteq Q$ is the set of accepting states. 
\end{definition}

The transition function $\delta$ is extended as $\delta(q, \sigma \cdot w) = \delta(\delta(q,\sigma), w)$. %
A word $w = w_0 w_1 \ldots w_n \in \Sigma^\ast$ is accepted by $\calA$ if and only if $\delta(\init, w)\in F$. 
The set of words accepted by $\calA$ is called the language of $\calA$, denoted by $\mathcal{L}(\calA)$.
Formally, $\mathcal{L}(\calA) = \{ w \in \Sigma^* \mid \delta(q_0, w) \in F \}$.

For notation simplicity, let $\Sigma\coloneqq 2^\calAP$.  The \ac{dfa} that determines P1's winning plays is defined as follows.
\begin{definition}
    Let $\mathcal{A} = ( Q, \Sigma, \delta, \init, F )$ be a DFA that specifies P1's temporal objective, i.e., $\varphi = \mathcal{L}(A)$. A play $\rho \in \plays(G)$ is called \emph{winning} for P1 if $L(\rho) \in \mathcal{L}(\mathcal{A})$. The set of all winning plays for P1 is denoted $\winplays_1$, i.e.,
\[
\winplays_1 = \{\rho \in \plays(G) \mid L(\rho) \in \mathcal{L}(\calA)  \}
\] 
\end{definition}

%

 In this paper, we assume the DFA that specifies the temporal goal has a specific structure.
\begin{assumption}
    All the accepting states of the DFA specifying the temporal goal of P1 are absorbing, that is, for each $q \in F$ and $\sigma \in \Sigma$, $\delta(q, \sigma) = q$.
\end{assumption}
%
 

 

 \begin{definition}
 Given a history $\rho \in \prefplays(G)$, 
     a strategy $\pi_1 \in \Pi_1$ is  \emph{winning} for P1 starting from $\rho$ if for any $\pi_2 \in \Pi_2$, $\plays(\rho, \pi_1, \pi_2) \subseteq \winplays_1$. 
     A strategy $\pi_2 \in \Pi_2$ is  \emph{winning} for P2 if for any $\pi_1 \in \Pi_1$, $\plays(\rho, \pi_1, \pi_2) \cap \winplays_1 = \emptyset$.
 \end{definition}

\subsection{Problem formulation}
 In this two-player game with temporal objective, we consider one single attacker of the system who is assumed to have full knowledge about the game arena but can only partially observe the plays in the game. The attacker's task is to infer if a play satisfy a secret property using his partial observations.
 
\begin{definition}[Attacker's observation function]
Given a finite set of observations $\Omega$, the attacker's observation function is a function $O: S\times A\times S \rightarrow \Omega$ that maps each transition $(s, a, s')$ to an observation the attacker receives when the transition system takes the transition $(s, a, s')$. The observation function naturally extends to plays: For each $\rho =s_0a_0s_1a_1\ldots s_n \in \plays(G)$, $O(\rho) = O(s_0,a_0,s_1)O(s_1,a_1,s_2)\ldots O(s_{n-1},a_{n-1},s_n) $.
Two plays $\rho_1,\rho_2$ are observation-equivalent if and only if $\obs(\rho_1)=\obs(\rho_2)$. We denote by $[\rho]$ the set of observation-equivalent plays of $\rho$. 
\end{definition}

Note this observation function is general enough to capture partial observations for both actions and states and also allows the action observations to be state-dependent. For example, we can capture the case when the attacker can observe P2's actions but not P1's actions.

The following information structure is considered:
\begin{itemize}
    \item P1 and P2 have perfect observations.
    \item The observation function of the attacker is a common knowledge to all players (P1, P2, and the attacker).
    \item There is no direct communication from P2 to the attacker. That is, P2 cannot inform the attacker if the secret is satisfied or not.
\end{itemize} 

To define the attacker's objective, we adopt the language-based opacity~\cite{lin2011opacity}. 
\begin{definition}[Language-based opacity]
Given a temporal objective $\varphi$, called the secret of P1, the secret $\varphi$ is \emph{opaque} with respect to a play $\rho \in \plays(G)$ if and only if \begin{inparaenum}
\item $L(\rho) \in \varphi$; and 
\item there exists at least one observation-equivalent play $\rho' \in [\rho] $ such that $L(\rho') \not\in \varphi$.
\end{inparaenum}
\end{definition}



 Using this notion of opacity, we can identity the set of winning and opaque plays as follows. 
\begin{definition} 
P1's \emph{opacity-enforcing winning} plays is a set $\opaquewinplays_1$ of plays such that 
\[
\opaquewinplays_1 = \{\rho \in \winplays_1 \mid \exists \rho' \in [\rho], L(\rho') \not\in \mathcal{L}(\mathcal{A})\}.
\] P1's \emph{secret-revealing winning} plays is a set $\revealwinplays_1$ of plays such that 
\[
\revealwinplays_1 = \{\rho \in \winplays_1 \mid  \forall \rho' \in [\rho], L(\rho') \in \mathcal{L}(\mathcal{A}) \}
\] 
\end{definition}
 Clearly, $\revealwinplays_1\cap \opaquewinplays_1=\emptyset$.

We consider that P2 is colluded with the attacker. In addition to P2's objective which is to prevent P1 from satisfying the objective $\varphi$, when $\varphi$ is satisfied, P2 will react in a way trying to reveal this secret to the attacker. 
\begin{definition}
Given a prefix $\rho_0\in \prefplays(G)$,
a strategy $\pi_1 \in \Pi_1$ is   \emph{opacity-enforcing winning}   if for any strategy $\pi_2 \in \Pi_2$, $\plays(\rho_0, \pi_1, \pi_2) \subseteq\opaquewinplays_1$. 
A strategy $\pi_1\in \Pi_1$ is called \emph{winning and positive opacity-enforcing} if for any strategy $\pi_2\in \Pi_2$, $\plays(\rho_0, \pi_1, \pi_2) \subseteq\winplays_1$ and $\Pr^{\rho_0, \pi_1,\pi_2}(\rho \in \opaquewinplays_1)>0$ and $\Pr^{\rho_0, \pi_1,\pi_2}(\rho \in \winplays_1 \setminus \opaquewinplays_1)>0$ where $\Pr^{\rho_0, \pi_1,\pi_2}$ is the probabilistic distribution induced by the strategy profile $(\pi_1,\pi_2)$ given a prefix $\rho_0\in \prefplays(G)$. 
\end{definition}

The notion of \emph{positive opacity-enforcing} is related to \emph{probabilistic opacity} and is relevant  when the distribution of plays induced by players' strategies is described by a stochastic process. Note that though the game arena is deterministic, we do not restrict players' strategies to be deterministic.

%

%

Having the above definitions, we now present our problem statement.
\problem
Given a turn-based, deterministic game arena $G =  \langle S := S_1 \cup S_2, A := A_1 \cup A_2, T, s_0, \calAP, L  \rangle$, an \ac{dfa}  $\mathcal{A} = ( Q, \Sigma, \delta, \init, F )$ describing P1's secret temporal goal, and an observation function $O: S\times A\times S \rightarrow \Omega$ describing an attacker's partial observations, for a given prefix $\rho_0$ from which P1 can ensure a winning play in $\winplays$, compute an opacity-enforcing winning strategy for P1, if exists. 
Otherwise, 
determine whether P1 has a winning and positive opacity-enforcing strategy.

\section{Main Results}
In this section, we describe our algorithm to solve an opacity-enforcing winning strategy for P1 through the construction of a belief-augmented game. Then we prove the correctness and completeness of the algorithm. In the end, we prove that P1 does not have a winning and positive opacity-enforcing  strategy in the turn-based deterministic game arena regardless if deterministic/non-deterministic/randomized strategies can be used.
%

 \begin{definition}[Belief-augmented game  arena]
\label{def:belief-augmented-arena}
Given the two-player, turn-based game arena $G =  \langle S := S_1 \cup S_2, A := A_1 \cup A_2, T, s_0, \calAP, L  \rangle$, and a \ac{dfa} $\mathcal{A} = (Q,2^\calAP, \delta, \init, F)$ describing P1's   secret objective $\varphi$, 
the belief-augmented game arena is a tuple 
\[
\mathcal{G} = (V := V_1 \cup V_2, A_1\cup A_2, \Delta, v_0)
\]
in which
\begin{itemize}
    \item   $V = S\times Q \times 2^{S\times Q}$ is the state space, partitioned into P1's states $V_1 =  S_1\times Q\times 2^{S\times Q}$ and P2's states $V_2= S_2\times Q\times 2^{S\times Q}$, where each state $(s,q,b)$ includes a game state $s\in S$, an automaton state $q\in Q$, and a belief state $b\subseteq S\times Q$ of the attacker;
    \item $A_1\cup A_2$ are the players' actions, same as in $G$;
    \item $\Delta: V \times (A_1\cup A_2) \rightarrow V$ is the transition function  such for a state $(s,q,b)\in V_i$, 
    and an action $a \in A_i$, 
    \[
    \Delta((s,q,b), a) = (s', q', b')
    \]
    where $s'=T(s,a)$, $q'=\delta(q, L(s'))$, and $b'= \{(s^o,q^o) \mid \exists (\bar s, \bar q)\in b, \exists \bar a \in A: \; \obs(\bar s, \bar a, s^o) = \obs(s, a, s') \text{ and } T(\bar s, \bar a)= s^o \text{ and } \delta(\bar q, L(s^o))= q^o\}$
\item $v_0 = (s_0,q_0,b_0)$ where $q_0 = \delta(\init, L(s_0))$ and $b_0  = \{(s_0,q_0)\}$, is the initial belief of the observer. 
\end{itemize}
\end{definition} 

Each state of this belief-augmented game is a tuple $(s, q, b)$ in which $s$ indicates the current state of the original arena, $q$ indicates the state of the specification DFA, to which the DFA reaches by tracing the label sequence of the current play, and $b$ is the belief of the attacker about the current state of the arena and the current state of the specification DFA.
%
All definitions related to  game arena in Def.~\ref{def:arena}, including those of a play and a strategy, are applicable to the belief-augmented game arena $\mathcal{G}$.

We can establish a one-to-one mapping $\mathfrak{R}: \plays(G)\rightarrow \plays(\mathcal{G})$ between plays in the original game $G$ and plays in the belief-augmented game $\mathcal{G}$ as follows:
For a play $\rho= s_0a_0s_1a_1\ldots s_n \in \plays(G)$, there is a unique play $\mathbf{p} = (s_0, q_0, b_0)a_0(s_1, q_1, b_1) \ldots (s_n, q_n, b_n) \in \plays(\mathcal{G})$
 where $q_0=\delta(\init, L(s_0))$ and $q_i  = \delta(q_{i-1}, L(s_{i-1}))$ for any $1\le i <n$. The belief state $b_i$ is constructed according to Definition~\ref{def:belief-augmented-arena} for $0 \le i \le n$. By construction, for a given belief $b_i$, the observation $\obs(s_i,a_i,s_{i+1})$, the new belief $b_{i+1}$ is uniquely determined \footnote{This is because the observation function is deterministic. This is not the case when the observation function becomes probabilistic.}. 

The following property can be shown.

\begin{lemma}
\label{lma:belief-property1}
Let $\mathbf{p} = (s_0, q_0, b_0)a_0(s_1, q_1, b_1) \ldots (s_n, q_n, b_n)$ be a play over $\mathcal{G}$. For all $i = 0, \dots, n$, it holds $(s_i,q_i)\in b_i$.
\end{lemma}
\begin{proof}
 We prove by induction on the length of $\mathbf{p}$. For $i = 0$, $b_0 = (s_0,q_0)$ by construction. 

Assume that $(s_k, q_k) \in b_k$ where $1 \leq k \leq n - 1$. Then by Def.~\ref{def:belief-augmented-arena}, for $i = k+1$, $b_{k + 1} = \{(s^o,q^o)\mid \exists (\bar s,\bar q)\in b_k, \bar a \in A: \obs(\bar s, \bar a, s^o) = \obs(s_k,a_k,s_{k+1}) \text{ and } T(\bar s, \bar a)= s^o \text{ and } \delta(\bar q, L(s^o))= q^o\}$. Since $(s_k, q_k) \in b_k$ and there is $\bar a = a_{k} \in A$ such that $s_{k + 1} = T(s_k, \bar a)$ and $\obs(s_k, \bar a, s_{k+1}) = \obs(s_k, a_k, s_{k+1})$, and that $\delta(q_k, L(s_k))= q_{k + 1}$, it holds that $(s_{k+1}, q_{k+1}) \in b_{k+1}$. Therefore, $(s_i,q_i) \in b_i$ for $i = 0,\dots,n$.
\end{proof}

This lemma 
states that the attacker's belief always contains the true state of the game and the current state of the specification DFA on tracing the label of the play.

Next, we show that two observation-equivalent plays yield, at each instant, equal beliefs about the status of the game.

\begin{lemma}
\label{lma:belief-property2}
Given a play $\rho = s_0a_0s_1a_1s_2 \dots s_n \in \plays(G)$, for any of its observation-equivalent play $\rho' \in [\rho]$, assuming $$\mathfrak{R}(\rho) = (s_0, q_0, b_0)a_0(s_1, q_1, b_1) \ldots (s_n, q_n, b_n)$$ and  $$\mathfrak{R}(\rho') = (s_0',q_0',b_0')a_0'(s_1',q_1', b_1') \dots (s_n', q_n', b_n'),$$ it holds that $b_i'  = b_i$ for all $0 \leq i \leq n$. 
\end{lemma}
\begin{proof}
For any $\rho' \in [\rho]$ given $\rho = s_0$, it holds that $\rho' = s\in [s_0]$. By definition, $b_0$ is a function of the set $[s_0]$ and $b'_0$ is a function of the set $[s] = [s_0]$. Thus, $b_0' = b_0$.

Assume that $b_k' = b_k$ where $1 \leq k \leq n - 1$. The belief update gives the new observation $\obs(s_k,a_k,s_k')$ based on the current belief $b_k$. The hypothesis $b_k=b_k'$ and $\rho'\in [\rho]$ implies $\obs(s_k',a_k',s_{k+1}') = \obs(s_k, a_k, s_{k+1})$. By the definition, $b_{k+1}$ is a function of $b_k$, $\obs(s_k, a_k, s_{k+1})$ and $b_{k+1}'$ is a function of $b_k'$ and $\obs(s_k',a_k',s_{k+1}')$. Since $\obs(s_k', a_k', s_{k+1}') = \obs(s_k,a_k,s_{k+1})$, we obtain $b_{k+1} = b_{k+1}'$. Therefore, we have $b_i' = b_i$ for all $1 \leq i \leq n$ by induction. 
\end{proof} 


%
\begin{lemma}
\label{lma:plays_and_belief}
    For any play $\rho =s_0s_1\ldots s_n \in \plays(G)$, let $\mathbf{p} = \mathfrak{R}(\rho)= (s_0, q_0, b_0)a_0(s_1, q_1, b_1) \ldots (s_n, q_n, b_n)$ be the corresponding play in the belief-augmented game $\mathcal{G}$.
    \begin{itemize}
        \item   $\rho \in \opaquewinplays_1$ if and only if $b_n \cap (S\times F) \ne \emptyset $ and $b_n\cap (S\times (Q\setminus F)) \ne \emptyset$.
        \item  $\rho \in \revealwinplays_1$ if and only if $b_n \subseteq (S\times F) $.
    \end{itemize}
\end{lemma}
\begin{proof}
 $\Rightarrow$:    First,  for all $i\ge 0$, $(s_i,q_i)\in b_i$. Thus, if $\rho\in \opaquewinplays_1$, then $q_n\in F$ and thus $b_n \cap S\times F\ne \emptyset$.

In addition, there exists $\rho'\in [\rho]$ such that $L(\rho')\not\models \varphi$. Let $\mathfrak{R}(\rho') = (s_0',q_0',b_0')a_0' \dots (s_n', q_n', b_n')$ be the corresponding play in $\mathcal{G}$ of $\rho'$. By Lemma~\ref{lma:belief-property1}, it holds that $q_n'\in Q\setminus F$ and thus $b_n' \cap S\times (Q\setminus F)\ne \emptyset$.

Finally, because $\rho$ and $\rho'$ are observation equivalent, then $b_i= b_i'$ for all $0\le i\le n$ by Lemma \ref{lma:belief-property2}. Thus, $ b_n' \cap S\times (Q\setminus F)\ne \emptyset$ is equivalent to $ b_n \cap S\times (Q\setminus F)\ne \emptyset$.

$\Leftarrow$: Conversely, if $b_n \cap S\times F\ne \emptyset $ and $b_n\cap S\times (Q\setminus F) \ne \emptyset$, then let $(s_n,q_n)\in b_n \cap S\times F$ and $(s_n', q_n')\in b_n \cap S\times (Q\setminus F)$. By the construction of the belief, there exists a play $\mathbf{p}' = (s'_0,q'_0,b_0) \ldots (s'_n,q'_n,b_n)$ such that $\mathfrak{R}^{-1}(\mathbf{p}') \in [\rho]$. Because $(s'_n, q_n')\in b_n \cap S\times (Q\setminus F) $, then $L(s_0's_1'\ldots s_n')\not\models \varphi$. 

The proof of the second statement follows analogously and thus is omitted.


\end{proof}

\begin{definition}[Reachability game and winning region \cite{gradelAutomataLogicsInfinite2002}]
Given the two-player turn-based game arena $\mathcal{G} =\langle V,A_1\cup A_2, \Delta, v_0\rangle$, a reachability objective, denoted $\Eventually X$ for $X\subseteq V$, can be satisfied by a play $\mathbf{p} =v_0v_1\ldots v_n$ if there exists $0\le i\le n$, $v_i\in \mathcal{F}_O$.
Given 
a reachability objective $\Eventually X$ for P1, the sure-winning region of P1 is a set $\win_1(X)$ of states in $\mathcal{G}$ from which P1 has a strategy to enforce a play that satisfies $\Eventually X$, regardless of the counter-strategy for P2. 
\end{definition}

\begin{theorem}
\label{thm1}
Given the original game arena $G =(S, A_1\cup A_2, T, s_0)$, P1's temporal logic objective $\varphi$ and its corresponding \ac{dfa} $\mathcal{A} = (Q,2^\calAP, \delta, \init, F)$,     P1 has an opacity-enforcing winning strategy in the original game $G$ with the temporal objective $\varphi$ if and only if P1 has a sure-winning strategy in the belief-augmented game arena $\mathcal{G}$ with a reachability objective $\Eventually \mathcal{F}_O$ \footnote{A reachability objective is equivalently expressed as $\Eventually p$, reads ``eventually $p$ is true'' and let atomic proposition $p\in L(q)$ for each $q\in \mathcal{F}_O$.}
where $$\mathcal{F}_O = \{(s,q,b)\mid q\in F, b \cap (S\times (Q\setminus F )) \ne \emptyset \}.$$
\end{theorem}


\begin{proof}
If P1 has an opacity-enforcing winning strategy $\pi_1^*$ in the original game $G$, then for any strategy $\pi_2$ of P2, for any play $\rho \in  \plays(s_0,\pi_1^\ast, \pi_2)  $, $\rho \in \opaquewinplays_1$. 
Let $\mathbf{p} = \mathfrak{R}(\rho)= (s_0, q_0, b_0)a_0(s_1, q_1, b_1) \ldots (s_n, q_n, b_n)$. By lemma \ref{lma:plays_and_belief}, $q_n\in F$ and $b_n\cap (S\times (Q\setminus F) ) \ne \emptyset$. In other words, the strategy $\pi_1^\ast$ enforces a run into $\mathcal{F}_O$ regardless of P2's strategy. Thus, this strategy $\pi_1^\ast$ is a sure-winning strategy in the game $\mathcal{G}$ for the reachability objective $\Eventually \mathcal{F}_O$.

Conversely, if P1 has a sure-winning strategy $\pi_1^{b*}$ in the belief-augmented game $\mathcal{G}$. Then for any strategy $\pi_2$ of P2, for any play $\mathbf{p} = (s_0,q_0,b_0)a_0 \dots (s_n, q_n, b_n) \in \plays(\mathcal{G}, \pi_1^{b*},\pi_2)$, the last state $(s_n, q_n, b_n)$ satisfies $q_n \in F$ and $b_n\cap (S\times (Q\setminus F) ) \ne \emptyset$. It is also noted that by Lemma~\ref{lma:belief-property1}, if  $q_n \in F$, then $b_n\cap (S\times F)) \ne \emptyset$. According to lemma \ref{lma:plays_and_belief}, $\rho = \mathbf{p} \in  \opaquewinplays_1$. Therefore, P1 has an opacity-enforcing  winning strategy in the original game $G$. 
 \end{proof}

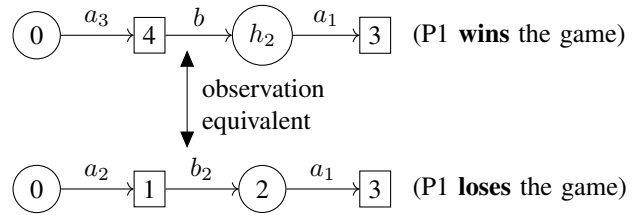
\begin{figure}[tp!]
  \centering
  
\begin{tikzpicture}[node distance=2cm, auto]

  \node[circle, draw, black] (0) at (0,0) {0};
  \node[rectangle, draw, black] (3) at (4.5,0) {3};
  \node[rectangle, draw, black] (4) at (1.5,0) {4};
  \node[circle, draw, black] (h2) at (3,0) {$h_2$};
  \node[text width=3cm] at (6.5,0) {(P1 \textbf{wins} the game)};

  \node[text width=2cm] at (3.2,-0.95) {observation equivalent};
  
  \draw[->, black] (0) -- node[above] {$a_3$} (4);
  \draw[->, black] (4) -- node[above] {$b$} (h2);
  \draw[->, black] (h2) -- node[above] {$a_1$} (3);
  \draw[>=triangle 45, <->] (2, -0.2) -- (2,-1.5);

\end{tikzpicture}

\begin{tikzpicture}[node distance=2cm, auto]

  \node[circle, draw, black] (0) at (0,0) {0};
  \node[rectangle, draw, black] (3) at (4.5,0) {3};
  \node[rectangle, draw, black] (1) at (1.5,0) {1};
  \node[circle, draw, black] (2) at (3,0) {2};
  \node[text width=3cm] at (6.5,0) {(P1 \textbf{loses} the game)};
  
  \draw[->, black] (0) -- node[above] {$a_2$} (1);
  \draw[->, black] (1) -- node[above] {$b_2$} (2);
  \draw[->, black] (2) -- node[above] {$a_1$} (3);

\end{tikzpicture}

\caption{An example of opacity-enforcing winning play (there exists a observation equivalent play which loses the game).}

\label{fig:illustrative_example}
\end{figure}

Since Theorem~\ref{thm1} shows the opacity-enforcing winning strategy in a deterministic game can be computed by solving a zero-sum turn-based deterministic game with reachability objective for P1 (see \cite{Mazala2002} for a  tutorial on the solutions of infinite games). The following statements can be derived from the solution concepts of zero-sum reachability games\cite{martinBorelDeterminacy1975,gradelAutomataLogicsInfinite2002}:
\begin{itemize}
    \item Memoryless, deterministic strategies are sufficient for sure-winning for both P1 and P2.     
    \item The game is determined: For any state $v\in V$, either P1 has a sure-winning strategy or P2 has a sure-winning strategy. 
    \item The sure-winning region for player $i$ is the same as the almost-sure winning region, which including all states where player $i$ has a strategy to ensure his/her objective to be satisfied with probability one. 
\end{itemize}
Thus, we refer to the sure-winning/almost-sure winning region as the winning region.

We introduce the winning regions in the belief-augmented game $\mathcal{G}$:
\begin{itemize}
    \item $\win_1^{opaque} = \win_1(\mathcal{F}_O)$ is the opacity-enforcing winning region for P1 (the sure-winning region for P1's objective $\Eventually \mathcal{F}_O$.)
    \item $ \win_1 = \win_1( \mathcal{ F})$, where 
    $\mathcal{ F}= S\times F \times 2^{S\times Q}$,
    is the winning region for P1 given the objective to satisfy $\varphi$ without enforcing opacity (the sure-winning region for P1's objective $\Eventually \mathcal{F}$.)
\end{itemize}

 The question is for a play $\rho \in \prefplays(G)$, if P1 can enforce satisfying the secret goal $\varphi$ but   cannot enforce opacity with certainty, does P1 has a strategy to enforce opacity with some positive probability?  The following theorem provides a negative answer to this question.

\begin{theorem}
\label{thm2}
Given the belief-augmented game $\mathcal{G}$, let $\widetilde \Pi_1 $ be the set of winning strategies for P1 with respect to the reachability objective $\Eventually \mathcal{F}$. For any state $v\in \win_1$, if P1 is restricted to policies in $\widetilde\Pi_1 $, then one of the two cases occurs:
\begin{itemize}
    \item If $v\in \win_1^{opaque}$, then there exists a strategy $\pi_1^o \in \widetilde \Pi_1$ that ensures opacity-enforcing winning.
    \item otherwise, for any strategy $\pi_1 \in \widetilde \Pi_1 $, any P2's strategy  $\pi_2$ ensure to reveal the secret when the secret is satisfied.
\end{itemize}
\end{theorem}
\begin{proof}
The first case is a natural consequence because $\mathcal{F}_O \subseteq  \mathcal{F} $. To prove the second case,
consider a state $v  = (s , q , b ) \in   \win_1 \setminus \win_1^{opaque}$, P2 has a    strategy $\pi^{\mbox{safe}}_2$ to enforce the set $\win_1^{opaque}$ is never reached, due to the duality of the reachability game. Let $\pi_1\in \widetilde \Pi_1$ and $\mathbf{p} =  v_0 v_1v_2\ldots v_n$ be a play starting from $v_0=v$, induced by the strategy profile $(\pi_1,\pi^{\mbox{safe}}_2)$, the following two conditions shall hold:
\begin{itemize}
    \item Since $\pi_1$ ensures $\mathcal{F}$ must be reached, there exists $k$, $ 0 \le k\le n$ such that $v_k=(s_k,q_k,b_k)\in \mathcal{F}$, which means $q_k\in F$ and $b_k \cap S\times F\ne \emptyset$ due to $(s_k,q_k)\in b_k$ (Lemma~\ref{lma:belief-property1}).
    \item Since $\pi^{\mbox{safe}}_2$ ensures a state in $\mathcal{F}_O$ is not reached, for any $j \ge 0$, $v_j=(s_j,q_j,b_j) \notin \mathcal{F}_O$, which means either $q_j\notin F$ or $b_j\cap S\times (Q\setminus F) = \emptyset$. 
\end{itemize}
Given that $q_k \in F$ followed from condition 1, then it is the case that 
$b_k \cap S\times F\ne \emptyset$ and, from condition 2,  
$b_k   \cap S\times (Q\setminus F) = \emptyset$. As a result, it must be the case that $b_k \subseteq S\times F$. Thus, the safety strategy played by P2 against P1 who only selects a winning strategy in $\widetilde \Pi_1$ ensures to reveal P1's secret.
\end{proof}

\begin{figure}[tp!]
\centering
\begin{tikzpicture}[->,>=stealth',shorten >=1pt,auto,node distance=2.5cm, scale=0.75,transform shape]
	\node[state] (0)  {\Large $0$};
	\node[state, rectangle, fill=blue!10] (5) [right=2cm of 0] {\Large $h_1$};
	\node[state, rectangle, fill=red!10] (1) [below=1.5cm of 5] {\Large $1$};
	\node[state, rectangle, fill=red!10] (4) [below=1.5cm of 0] {\Large $4$};
	\node[state, fill=blue!10] (2) [right=2cm of 5] {\Large $2$};
	\node[state, rectangle, double] (3) [right=2cm of 2] {\Large $3$};
	\node[state, fill=blue!10] (6)  [right=3cm of 1] {\Large $h_2$};
	
	\path 	
		(0) edge   node {\Large $a_1$} (5)
		(0) edge   node {\Large $a_2$} (1)
		(0) edge   node {\Large $a_3$} (4)
		(5) edge   node {\Large $b_1$} (2)
		(5) edge   node[pos=0.8] {\Large $b_2$} (6)
		(1) edge   node[pos=0.1] {\Large $b_2$} (2)
		(1) edge   node[above] {\Large $b_1$} (6)
		(2) edge   node {\Large $a_1$} (3)
		(6) edge   node {\Large $a_1$} (3)
		(4) edge[bend right]   node {\Large $a_1$} (6)
		;
\end{tikzpicture}
\caption{An illustrative example of opacity-enforcing winning.}
\label{fig:example1}
\end{figure}
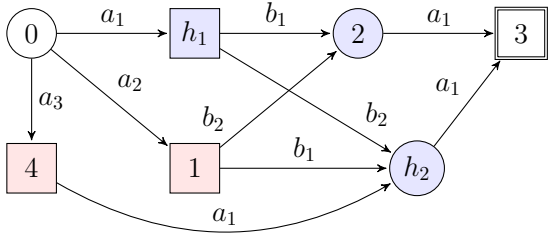

%
%
%
%
%


Based on Theorem~\ref{thm2}, from $v\in \win_1\setminus \win_1^{opaque}$  a P2's strategy to reveal P1's secret can be any strategy in the game in which P2 aims to prevent a play reaching $\win_1^{opaque}$. There can be two outcomes, either P1 chooses to not satisfy the goal by not committing to a strategy in $\widetilde \Pi_1$, or the secret will be revealed the moment a state in $\cal F$ is reached.

\begin{lemma}
    The time complexity of our algorithm to solve our problem with the game arena $G =(S, A_1\cup A_2, T, s_0)$ and the \ac{dfa} $\mathcal{A} = (Q,2^\calAP, \delta, \init, F)$ as inputs to the problem is $\mathcal{O}(|A||S|^2|Q|^22^{|S||Q|})$.
\end{lemma}
\begin{proof}
    The time of our algorithm consists of (1) the time to constrcut the belief-based game $\mathcal{G}$ and (2) the time to compute a winning strategy for P1 on $\mathcal{G}$.
    The belief-based game $\mathcal{G}$ has $\mathcal{O}(|V|)=\mathcal{O}(|S||Q|2^{|S||Q|})$ states and $\mathcal{O}(|\Delta|)=\mathcal{O}(|A||V|)$ transitions. %
    Using appropriate data structures, mainly hash-tables, it takes $\mathcal{O}(|A||S|^2|Q|^22^{|S||Q|})$ to construct the belief-based game. This is because for each state of this game, which is in the form of a $v=(s, q, b) \in V$, and for each action $a \in A$, we must compute what is $\Delta(v, a)$.
    Computing that requires to access $T(\bar s, a)$ and $\delta(\bar q, a)$ for each $(\bar s, \bar o) \in b$ and the size of $b$ can be $|S||Q|$ in the worst case.
    Using the Zielonka's algorithm~\cite{zielonka1998infinite}, it takes $\mathcal{O}(|V|+|\Delta|)=\mathcal{O}(|A||S||Q|2^{|S||Q|})$ to compute a winning strategy for P1 on $\mathcal{G}$.
    Therefore, the time complexity of our algorihtm is $\mathcal{O}(|A||S|^2|Q|^22^{|S||Q|})$.
\end{proof}

\section{Experiments}

\begin{example}
In this illustrative example (Fig \ref{fig:example1}), we demonstrate how to solve an opacity-enforcing winning strategy for P1. The game consists of seven states, with circles representing P1's states and squares representing P2's states. P1 can choose from three actions, $A_1 = \{a_1, a_2, a_3\}$, while P2 has two actions, $A_2 = \{b_1, b_2\}$. State 3 is a blocking state so that when $3$ is reached, the game ends.

The partial observations of the attacker are described as follows: The states $h_1,h_2,2$ are indistinguishable, and the states $1,4$ are indistinguishable. The actions chosen by both players cannot be observed by the attacker. The goal of P1 is to visit $h_1$ or $h_2$, \ie, to satisfy $\varphi = \Eventually (h_1 \lor h_2)$. The \ac{dfa} representing the goal is shown in Fig.~\ref{fig:DFA}. 

\begin{figure}[tp!]

\centering
\begin{tikzpicture}[shorten >=1pt,node distance=2.5cm,on grid,auto]
    \node[state,initial](q_0)  {$q_0$};
    \node[state,accepting]          (q_1) [right=of q_0] {$q_1$};
    
    \path[->] (q_0) edge node {$h_1\lor h_2$} (q_1)
    (q_0) edge [loop above] node {$\neg h_1\land \neg h_2$} (q_0)
    (q_1) edge [loop above] node {$\top$} (q_1)
    ;
\end{tikzpicture}
\caption{The DFA for $\Eventually (h_1\lor h_2)$.}

\label{fig:DFA}
\end{figure}

\begin{figure}[tp!]
\centering
\begin{tikzpicture}[->,>=stealth',shorten >=1pt,auto,node distance=2.5cm, scale=0.45,transform shape]
	\node[state,ellipse] (1) [] {\Large $(0, q_0,\{(0,q_0))\}$};
	\node[state,ellipse,fill=red!20] (2) [left=1.25cm of 1] {\Large $(h_1, q_1, \{(h_1,q_1)\})$};
	\node[state,ellipse] (3) [right=1.25cm of 1] {\Large $(4,q_0,\{(4,q_0)\})$};
	\node[state,ellipse] (4) [below right=2cm of 1] {\Large $(h_2,q_1,  \{(2,q_0),(h_2,q_1)\})$};
	\node[state,ellipse] (5) [below left=2cm of 1] {\Large $(3,q_1, \{(3,0),(3,q_1)\})$};

	\path 
	(1) edge[above]   node {\Large $a_1$} (2)
	(1) edge[above]   node {\Large $a_3$} (3)
	(3) edge[right]   node {\Large $b_1, b_2$} (4)
	(4) edge[above]   node {\Large $a_1$} (5)
	;
\end{tikzpicture}
\caption{The belief-augmented game for illustrative example.}
\label{fig:belief-aug-game}
\end{figure}
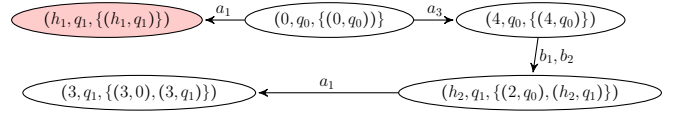


A fragment of the belief-augmented game is drawn in Fig.~\ref{fig:belief-aug-game}. A state, for example, $(h_2,q_1,\{(2,q_0), (h_2,q_1)\})$ means the game state is $h_2$, the \ac{dfa} state is $q_1$ and the attacker's belief is $\{(2,q_0), (h_2,q_1)\}$. The action $a_2$ is omitted. It leads to a losing state because P2 can select action $b_2$ at $1$ and thus P1 cannot satisfy the objective. 

Based on the solution of the reachability game, the following result is obtained: 
At state 0, P1 has an opacity-enforcing winning strategy. 
The strategy results in two possible plays of the form $0 \xrightarrow{a_3} 4 \xrightarrow{b} h_2 \xrightarrow{a_1} 3$, where $b\in\{b_1,b_2\}$. The observations made during these plays are $0 \{1,4\} \{h_1,h_2,2\} 3$. By applying the inverse of the observation function to this observation, we can identify a play $\rho' = 0 \xrightarrow{a_2} 1 \xrightarrow{b_2} 2 \xrightarrow{a_1} 3$ such that $L(\rho') \not\in \mathcal{L}(\calA)$ and $\obs(\rho') = 0 \{1,4\} \{h_1,h_2,2\} 3 = \obs(\rho)$. Thus, this play $\rho$ is opaque and winning for P1. 

It is noted that in this game, if opacity is not enforced, P1 can satisfy his objective by taking action $a_1$ or $a_3$   from state $0$. 
 However, action $a_1$ can lead to a state outside of $\win_1^{opaque}$ in the belief-augmented game because by taking $a_1$, state $h_1$ will be reached. And two runs are possible: $0\xrightarrow{a_1}h_1\xrightarrow{b_1}2\xrightarrow{a_1}3$ or $0\xrightarrow{a_1} h_1\xrightarrow{b_2}h_2\xrightarrow{a_1}3$. The observations are $0,\{h_1,h_2,2\},\{h_1,h_2,2\}, 3 $ for both runs. Regardless of P2 taking action $b_1$ or $b_2$, the attacker can infer at least $h_1$ is visited and the formula $\varphi$ is satisfied.

\end{example}

\begin{example}
Consider the $5 \times 5$ grid-world in Fig.~\ref{fig:grid_world}. 
In this example, the aerial robot, P1, and the ground robot, P2, have $4$ movement actions $\{N, S, W, E\}$ to navigate the environment.
P1 can navigate the whole area, while P2 can only navigate the area in gray and it cannot enter the water area, the $9$ cells in the middle.
P1's temporal goal is described by the temporal logic 
formula $\varphi = \Eventually (G_0 \land \Eventually (G_1 \lor G_2))$, 
stating first visit $G_0$ and then $G_1$ or $G_2$.
If the robots move into the same cell, P2 can intercept P1, and the game ends.
%
%
The game is turn-based, and at each step, only one robot makes a move.
Besides the two robots, there is an observer to whom the current state of the game is only partially observable. 
The observer obtains information from two kinds of  sensors in the environment,
\begin{itemize}
    \item[1.] The $5$ range sensors where each guards a row of the gridworld (when P1 enters a row, the observer knows at which row P1 is located, but it does not know in which column P1 is); and 
    \item[2.] The cross-shaped range sensor is shown in Figure~\ref{fig:grid_world}b, which is carried by robot P2. Once robot P1 enters the range of the sensor, the observer will receive the exact location of P1.
\end{itemize}
The goal of P2 is to intercept P1 if it can or reveal P1's secret to the observer if it cannot prevent P1 from winning the game.
%

\begin{figure}[htp!]
\includegraphics[width=\linewidth]{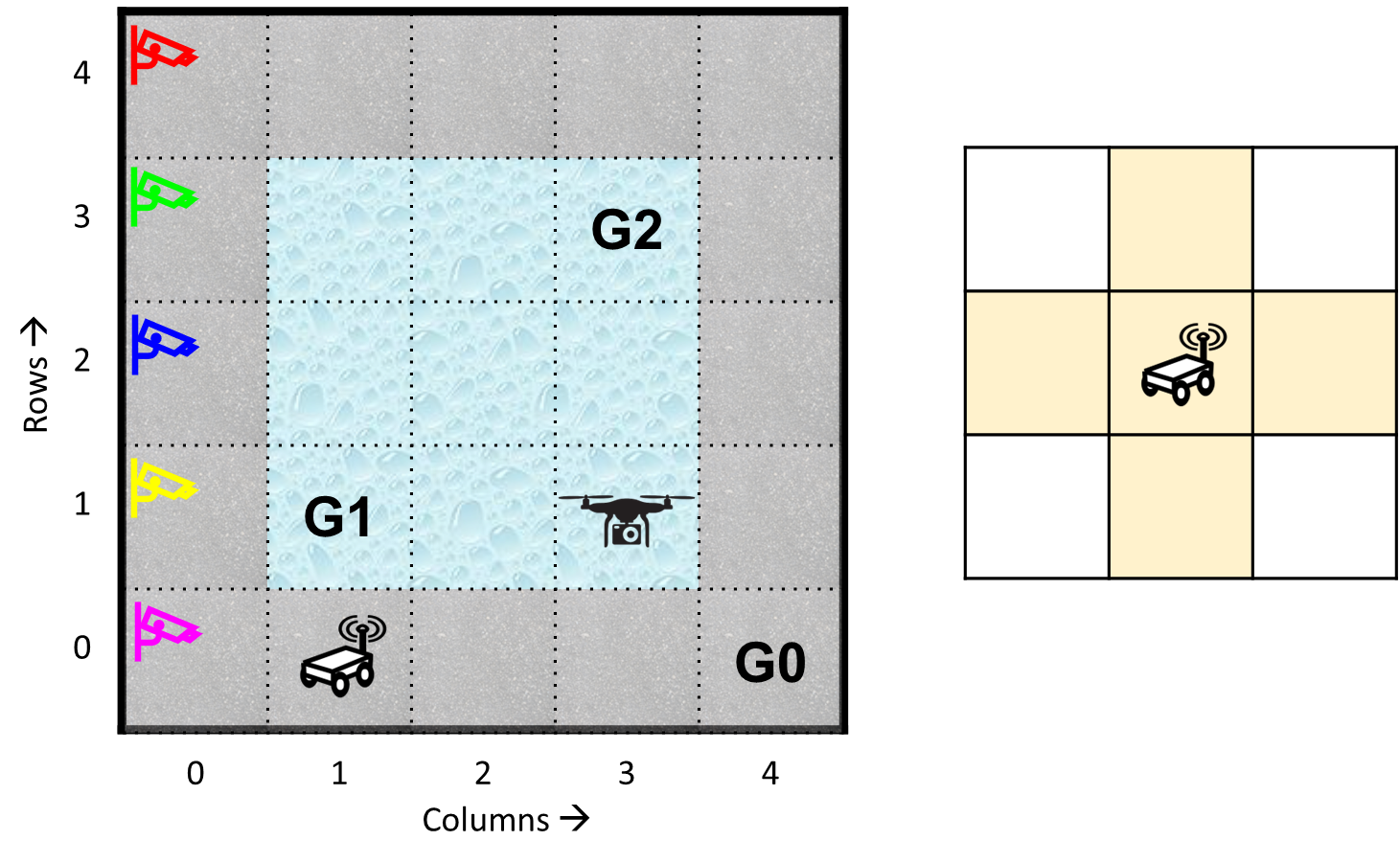}
\centering
\caption{\textbf{(left)} A robot motion planning problem motivating our opacity-enforcing temporal planning. The aerial robot's temporal goal is to first collect a sample at $G_0$ and then collect a sample at either $G_1$ or $G_2$. Five range sensors guard the environment, where each tells the observer the P1 is in that row whenever P1 enters that row. The range sensors do not know the column at which P1 might be located. \textbf{(right)} The cross-shaped range sensor is carried by the ground robot.}
\label{fig:grid_world}
\end{figure}

 %
For this example, we considered $25$ different instances for the $25$ possible initial locations of the aerial robot, P1. In each instance, we fix the initial location of the ground robot, P2, at the bottom-left cell $(0,0)$ of the environment.
We used our algorithm to solve these $25$ instances of our problem.
Figure~\ref{fig:result} shows the results of our algorithm.
The cells that contain red circles show the P1's initial locations from which P2 has a winning strategy to prevent P1 from satisfying $\varphi$.
Cells with a yellow circle show P1's initial locations from which, if the game starts, P1 has a winning strategy to satisfy $\varphi$ but cannot enforce opacity.
The green circles are shown in the cells from which, if the game starts, P1 has an opacity-enforcing winning strategy. 
Accordingly, the green and yellow cells are winning initial states for P1. But if P1 wants to enforce opacity, then only green initial states allow it to do so.

\begin{figure}
    \centering
    \includegraphics[scale=0.4]{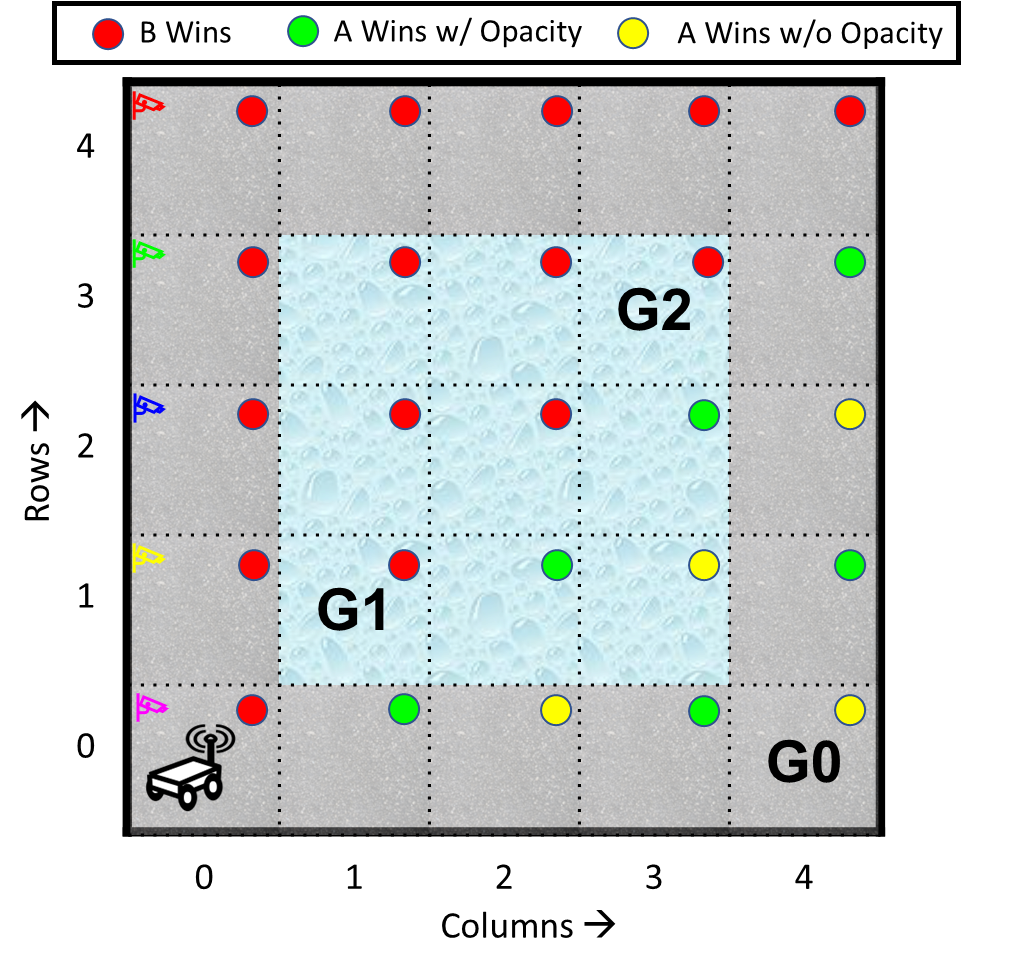}
    \caption{The results of our experiment for the $25$ instances of the grid-world example, where P2's initial state is fixed at $(0, 0)$ in all those instances, and P1's initial state varies across the instances.
    }
    \label{fig:result}
\end{figure}

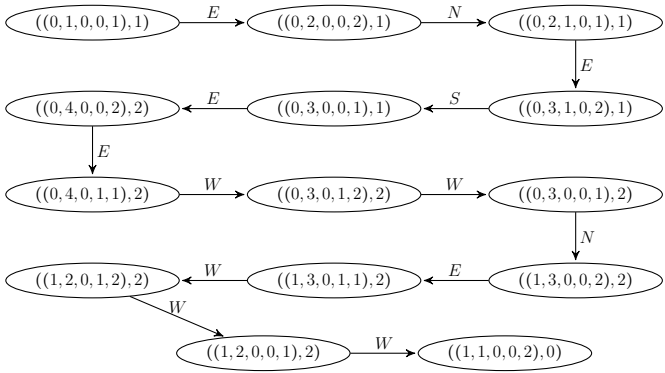
\begin{figure}[t]
	\begin{center}
		\begin{tikzpicture}[->,>=stealth',shorten >=1pt,auto,node distance=2.5cm, scale = 0.45,transform shape]
			\node[state,ellipse] (1) [] {\Large $((0, 1, 0, 0, 1), 1)$};
			\node[state,ellipse] (2) [right=2cm of 1] {\Large $((0, 2, 0, 0, 2), 1)$};
			\node[state,ellipse] (3) [right=2cm of 2] {\Large $((0, 2, 1, 0, 1), 1)$};
			\node[state,ellipse] (4) [below=1.5cm of 3] {\Large $((0, 3, 1, 0, 2), 1)$};
			\node[state,ellipse] (5) [left=2cm of 4] {\Large $((0, 3, 0, 0, 1), 1)$};
			\node[state,ellipse] (6) [left=2cm of 5] {\Large $((0, 4, 0, 0, 2), 2)$};
			\node[state,ellipse] (7) [below=1.5cm of 6] {\Large $((0, 4, 0, 1, 1), 2)$};
			\node[state,ellipse] (8) [right=2cm of 7] {\Large $((0, 3, 0, 1, 2), 2)$};
			\node[state,ellipse] (9) [right=2cm of 8] {\Large $((0, 3, 0, 0, 1), 2)$};
			\node[state,ellipse] (10) [below=1.5cm of 9] {\Large $((1, 3, 0, 0, 2), 2)$};
			\node[state,ellipse] (11) [left=2cm of 10] {\Large $((1, 3, 0, 1, 1), 2)$};
			\node[state,ellipse] (12) [left=2cm of 11] {\Large $((1, 2, 0, 1, 2), 2)$};
			\node[state,ellipse] (13) [below right=2cm of 12] {\Large $((1, 2, 0, 0, 1), 2)$};
			\node[state,ellipse] (14) [right=2cm of 13] {\Large $((1, 1, 0, 0, 2), 0)$};

			\path (1) edge[above]   node {\Large $E$} (2)
			(2) edge[above]   node {\Large $N$} (3)
			(3) edge[right]   node {\Large $E$} (4)
			(4) edge[above]   node {\Large $S$} (5)
			(5) edge[above]   node {\Large $E$} (6)
			(6) edge[right]   node {\Large $E$} (7)
			(7) edge[above]   node {\Large $W$} (8)
			(8) edge[above]   node {\Large $W$} (9)
			(9) edge[right]   node {\Large $N$} (10)
			(10) edge[above]   node {\Large $E$} (11)
			(11) edge[above]   node {\Large $W$} (12)
			(12) edge[above]   node {\Large $W$} (13)
			(13) edge[above]   node {\Large $W$} (14)
			;
		\end{tikzpicture}
	\end{center}
 \caption{Opacity-enforcing winning path given P1 starts from $(0,1)$.}
 \label{fig:winning-path}
\end{figure}
We have run $25$ experiments to solve the winning regions from $25$ different initial states on Intel (R) Core (TM) i7-5820K CPU @ 3.30GHz 3.30 GHz. The average time consumed for solving one game is 7.2 seconds. Fig.~\ref{fig:winning-path} shows the opacity-enforcing winning run when P1 starts from $(0,1)$ and P2's initial location is $(0, 0)$. We encode the state as $((p_{1r}, p_{1c}, p_{2r}, p_{2c}, turn), q)$ when\begin{itemize}
    \item P1's position is $(p_{1r}, p_{1c})$,
    \item P2's position is $(p_{2r}, p_{2c})$,
    \item $turn = 1,2$ indicates the turn of the game (either P1's turn or P2's turn), 
    \item $q$ is the state of DFA. The accepting states $F=\{0\}$.
\end{itemize}
In the last state of the path, the belief state is $\{(1, 1, 0, 0, 2), 0), ((1, 3, 0, 0, 2), 2)\}$ that contains a state with $q = 2$. Since $q = 2$ is not in $F$, it is an opacity-enforcing winning play for P1. 
\end{example}

\section{Conclusion and Future Work}
In this work, we formulate and solve the control synthesis problem for an agent to satisfy its temporal goal while enforcing opacity against its uncontrollable environment player, who colludes with a passive observer. We prove that in a turn-based deterministic game with perfect observation between P1 and P2, there does not exist a winning and positively opacity enforcing strategy for P1. However, it remains to be answered whether a winning and positively opacity enforcing strategy may exist for concurrent/stochastic games or games with partial observations, or whether existing game objectives can capture infinite-state opacity \cite{yinInfinitestepOpacityKstep2019}.
For other classes of games and different information structures between players and observers, quantitative notions of opacity such as probabilistic  opacity \cite{berardQuantifyingOpacity2015,sabooriCurrentStateOpacityFormulations2014} can be further investigated, leveraging results from stochastic games.
 


\bibliographystyle{IEEEtranS}
\bibliography{references}

\end{document}